\documentclass[aps,pre,reprint,superscriptaddress,longbibliography]{revtex4-1}

\usepackage{etex}
\makeatletter
\def\alloc@#1#2#3#4#5%
 {\ifnum\count1#1<#4
    \allocationnumber\count1#1
    \global\advance\count1#1\@ne
    \global#3#5\allocationnumber
    \wlog{\string#5=\string#2\the\allocationnumber}%
  \else\ifnum#1<6
    \def\etex@dummy@definition{}
    \begingroup \escapechar\m@ne
    \expandafter\alloc@@\expandafter{\string#2}#5%
  \else\errmessage{No room for a new #2}\fi\fi
 }
\makeatother

\usepackage{etoolbox} 
\usepackage{amsmath,amssymb,amsfonts,amsthm}    
\usepackage{graphicx}
\usepackage{xcolor} 
\usepackage{easybmat}
\usepackage{booktabs}
\usepackage{multirow,bigdelim}
\usepackage{cases}

\usepackage{hyperref}
\hypersetup{pdfborder={0 0 0}} 

\usepackage[capitalize]{cleveref}
\makeatletter
\appto{\appendix}{%
  \@ifstar{\def\theequation@prefix{A.}}%
          {}%
}
\makeatother

\newtheorem{prop}{Proposition}[section]
\crefname{prop}{Proposition}{Propositions}
\crefname{definition}{Definition}{Definitions}
\crefname{assumption}{Assumption}{Assumptions}
\crefname{remark}{Remark}{Remarks}

\newtheorem{definition}{Definition}[section]
\newtheorem{assumption}{Assumption}[section]

\theoremstyle{remark}
\newtheorem{remark}{Remark}[section]

\newlength \figwidth
\begin{document}

\title{Controlling statistical moments of stochastic dynamical networks}

\author{Dmytro Bielievtsov}
\thanks{These authors contributed equally to this work.\\ 
E-mail: belevtsoff@gmail.com (D. Bielievtsov)\\
E-mail: josef.ladenbauer@tu-berlin.de (J. Ladenbauer)}
\affiliation{Bernstein Center for Computational Neuroscience Berlin, Philippstra\ss{}e 13, 10115 Berlin, Germany}

\author{Josef Ladenbauer}
\thanks{These authors contributed equally to this work.\\ 
E-mail: belevtsoff@gmail.com (D. Bielievtsov)\\
E-mail: josef.ladenbauer@tu-berlin.de (J. Ladenbauer)}
\affiliation{Institut f\"ur Softwaretechnik und Theoretische Informatik, Technische Universit\"at Berlin, Marchstra\ss{}e 23, 10587 Berlin, Germany}
\affiliation{Bernstein Center for Computational Neuroscience Berlin, Philippstra\ss{}e 13, 10115 Berlin, Germany}

\author{Klaus Obermayer}
\affiliation{Institut f\"ur Softwaretechnik und Theoretische Informatik, Technische Universit\"at Berlin, Marchstra\ss{}e 23, 10587 Berlin, Germany}
\affiliation{Bernstein Center for Computational Neuroscience Berlin, Philippstra\ss{}e 13, 10115 Berlin, Germany}

\date{\today}

\begin{abstract}  
    We consider a general class of stochastic networks and ask which network
    nodes need to be controlled, and how, to stabilize and switch between
    desired metastable (target) states in terms of the first and second
    statistical moments of the system. We first show that it is sufficient to
    directly interfere with a subset of nodes which can be identified using
    information about the graph of the network only. Then, we develop a
    suitable method for feedback control which acts on that subset of nodes and
    preserves the covariance structure of the desired target state. Finally, we
    demonstrate our theoretical results using a stochastic Hopfield network and
    a global brain model. Our results are applicable to a variety of (model)
    networks, and further our understanding of the relationship between network
    structure and collective dynamics for the benefit of effective control.
\end{abstract}

\pacs{}
\maketitle

\section{\label{Sec:Intro}Introduction}

Complex dynamical network models are widely applied in science and engineering
to describe and examine various natural and human-produced phenomena.  Examples
include models of neuronal networks, gene regulatory networks, financial
networks, social networks, power grids and highway webs
\cite{Newman2003,Boccaletti2006}. 
Noise is often included to reflect the variability observed in the real system; therefore, stochastic dynamical networks constitute a particularly
relevant model class. 

Of great interest are the relationships between the structural properties of
the network (i.e., properties of its graph) and the network dynamics -- in
particular, stochastic synchronization phenomena and their controllability
\cite{Glass2001,Neiman1999,Zaks2005,Zhou2005,Porfiri2008,Wang2010}.  In human
neuroscience, for example, the zero-lag correlations between measured
fluctuating (neuronal) activities of different brain areas reveal functional
connectivity patterns and can thereby serve as a biomarker for a number of
brain diseases \cite{Fox2005,Damoiseaux2006,Rosazza2011} These connectivity
patterns are strongly influenced, but not fully determined, by the brain
network graph \cite{Haimovici2013,Deco2013JN,Deco2013TINS}. 

In this contribution, we ask whether and how the first and second statistical
moments of stochastic dynamical networks can be controlled in general by
directly interfering only with a subset of nodes. These nodes should be
determined using solely information about the network graph.  Specifically, we
consider a general class of stochastic dynamical network models which possess
a set of collective equilibrium states in terms of steady-state mean values of
node variables and covariance patterns.  The aim is to switch between these
network states  by pinning only a subset of nodes in an appropriate way.  That
is, we consider pinning control (also known as clamping control) and examine
the capabilities of the pinned nodes to hold the means for the free nodes and
the covariances between them at the desired equilibrium values.

Our approach consists of three steps (see \cref{fig:control:full}): First, we
describe the dynamics of the first and second moments of the stochastic network
system by a closed system of coupled ordinary differential equations using a
method of moments (\cref{Sec:Models}).  Based on the graph of that
deterministic system we then identify subsets of moments, which, when
controlled, steer the dynamics of the whole moments system to a desired
metastable state (\cref{Sec:ControlNodes}).  Since the graph of the system that
describes the dynamics of the moments differs from the graph of the underlying
network system, we next determine the subset of nodes that need to be
controlled in the original network (same section).  In the last step we design
a feedback controller that acts on the determined subset of nodes of the
stochastic network system to switch between correlation patterns
(\cref{Sec:Controller}).  We demonstrate the applicability of our theoretical
results using two example models: i) a noisy Hopfield network -- a well-studied
class of recurrent neural network models (\cref{Sec:Results:sim}) and ii) a
calibrated global brain network (\cref{Sec:Results:sim2}). Limitations of the
developed method and possible adjustments are discussed in
\cref{Sec:Discussion}. 

\section{\label{Sec:Models}Stochastic network model and the moments system}

Consider the network model consisting of $N$ nodes described by
\begin{equation}
    \dot{\mathbf{x}}_i = \mathbf{f}_i(t, \mathbf{x}_i, \mathbf{x}_{I_i}) + \mathbf{M}_i \boldsymbol{\eta}_i(t),
    \label{eq:stochastic:components}
\end{equation}
where $\mathbf{x}_i(t) \in \mathbb{R}^{m_i}$ is the (time-varying)
$m_i$-dimensional state vector of node $i \in \{1,\dots,N\}$, the vector field
$\mathbf{f}_i : \mathbb{R}\times\mathbb{R}^{m_i}\times\mathbb{R}^{\sum_{j \ne
i}m_j}\to \mathbb{R}^{m_i}$ is continuous, uniformly bounded and sufficiently
differentiable, and $I_i$ is the set of input node indices for node $i$.
$\boldsymbol{\eta}_i(t)$ is a $m_i$-dimensional Gaussian white noise process
with $\langle\boldsymbol{\eta}_i(t)\boldsymbol{\eta}_i(t+\tau)^T\rangle =
\delta(\tau)\mathbf{I}_{i}$, where $\mathbf{I}_{i}$ is the $m_i\times m_i$
identity matrix, and  $\mathbf{M}_i$ is a constant noise mixing matrix.
Equation~\eqref{eq:stochastic:components} can be expressed in compact form as
\begin{equation}
    \dot{\mathbf{x}} = \mathbf{f}(t, \mathbf{x}) + \mathbf{M}\boldsymbol{\eta}(t),
    \label{eq:pt1:stochastic}
\end{equation}
where $\mathbf{x} \in \mathbb{R}^k, k:=\sum_i m_i$, $\mathbf{f}(t, \mathbf{x})
= [\mathbf{f}_1(t, \mathbf{x}_1, \mathbf{x}_{I_1})^T,\dots,\mathbf{f}_N(t,
\mathbf{x}_N, \mathbf{x}_{I_N})^T]^T$, $\mathbf{M}$ is a block matrix with
blocks $\mathbf{M}_i$ on the main diagonal, and $\boldsymbol{\eta}(t)$ is a
$k$-dimensional Gaussian white noise process.  $\mathbf{f}(t, \mathbf{x})$
generates a directed \emph{graph} $\Gamma$ with $N$ nodes, where the edges that
connect node $i$ with its input nodes are indicated by the set $I_i$.

\begin{figure}[!htbp]
    \includegraphics[width=\columnwidth]{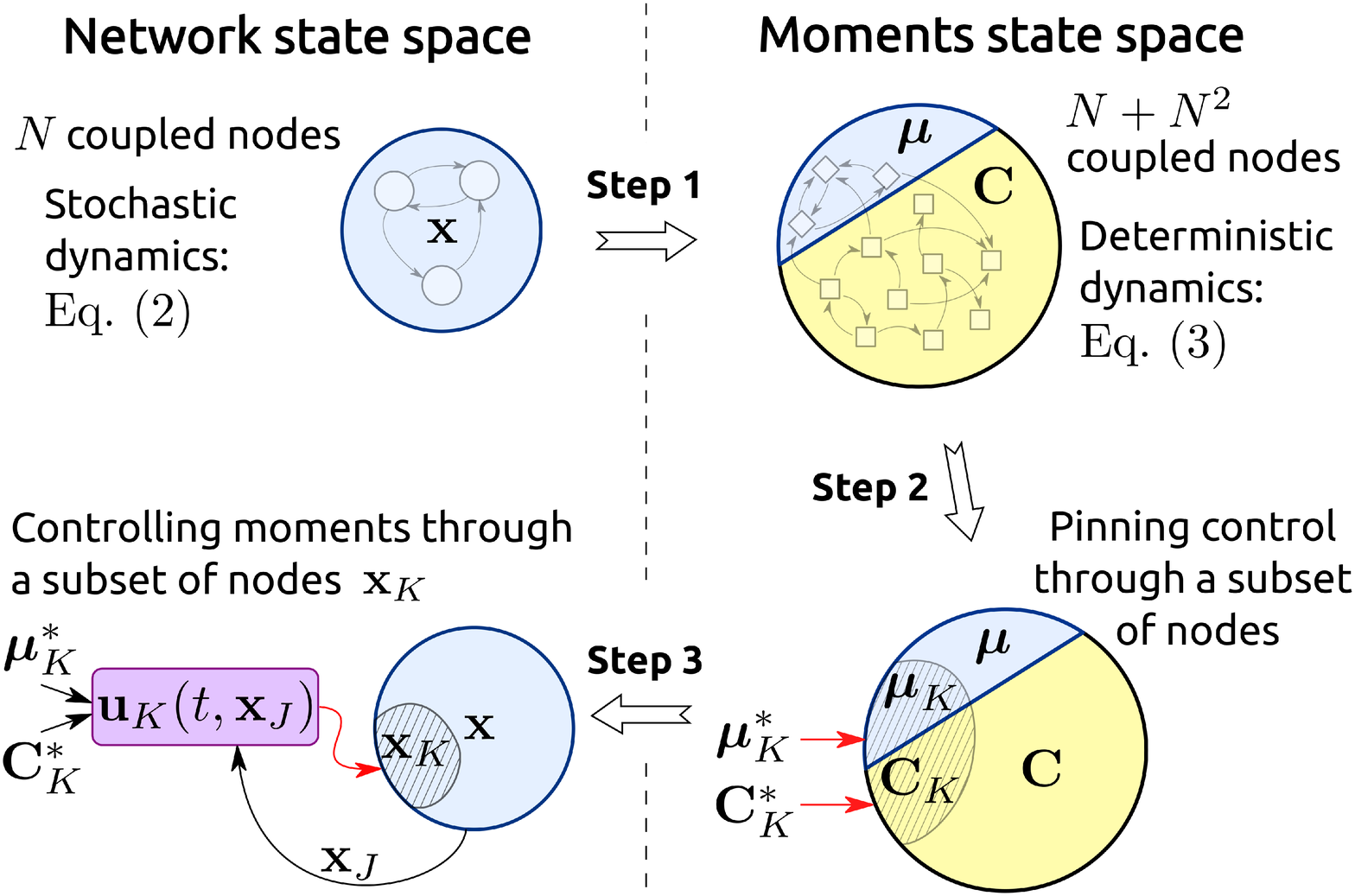}
        \caption{Schematic of the control scheme. Left: State
        space of the stochastic dynamical network whose state is represented by
        the vector $\mathbf{x}$.  Right: State space of the corresponding
        moments $\boldsymbol{\mu}$ (mean vector) and $\mathbf{C}$ (covariance
        matrix) of $\mathbf{x}$, whose dynamics is governed by a  deterministic
        system of equations.  Red arrows indicate pinning control through the
        subsets of nodes $\boldsymbol{\mu}_{K}$, $\mathbf{C}_{K}$ and
        $\mathbf{x}_K$, respectively.  $\boldsymbol{\mu}_K^*$ and
        $\mathbf{C}_K^*$ are the values of dynamic variables
        $\boldsymbol{\mu}_K$ and $\mathbf{C}_K$ on the (metastable) target state.
        $\mathbf{u}_K(t,\mathbf{x}_J(t))$ denotes the feedback control signal
        which directly affects the subset of nodes $\mathbf{x}_K$ and depends
        on the complementary subset $\mathbf{x}_J$.}
    \label{fig:control:full}
\end{figure}

The dynamics of the mean vector $\boldsymbol{\mu} (t) := \langle \mathbf{x}
\rangle (t)$ and the covariance matrix $\mathbf{C}(t) :=
\mathrm{cov}(\mathbf{x},\mathbf{x})(t) =
\langle(\mathbf{x}-\boldsymbol{\mu})(\mathbf{x}-\boldsymbol{\mu})^T \rangle
(t)$ for the system \cref{eq:pt1:stochastic} can be described by
(deterministic) ordinary differential equations using a method of moments
assuming weak noise \cite{Rodriguez1996}.  Employing the Fokker-Planck equation
for the probability density function $p(\mathbf{x}, t)$ of $\mathbf{x}$ at time
$t$, integrating that equation over the state space and using second order
Taylor expansion of $\mathbf{f}(t, \mathbf{x})$ with respect to $\mathbf{x}$
under the integral we obtain the moments system (MS),
\begin{subequations}
    \begin{align}
        \dot{\boldsymbol{\mu}} &= \mathbf{f}(t, \boldsymbol{\mu})+\frac{1}{2}\nabla\nabla \mathbf{f}(t, \boldsymbol{\mu})\cdot \mathbf{C} \label{eq:pt1:mu} \\
        \dot{\mathbf{C}} &= \nabla \mathbf{f}(t, \boldsymbol{\mu})\mathbf{C} + \mathbf{C}\nabla \mathbf{f}(t, \boldsymbol{\mu})^T 
        + \mathbf{Q} \label{eq:pt1:c},
    \end{align}
    \label{eq:pt1:muc}%
\end{subequations}
where $\nabla \mathbf{f}(t, \mathbf{x})$ is the Jacobian matrix of $\mathbf{f}$
with respect to $\mathbf{x}$ and $\mathbf{Q}:=\mathbf{M}\mathbf{M}^T$ is the
total noise covariance matrix. Note, that
\begin{equation*}
    {\left[\nabla\nabla \mathbf{f}(t, \boldsymbol{\mu}) \cdot \mathbf{C}\right]_{j} = \sum_{l, p}\frac{\partial^2}{\partial x_l\partial x_p}
    \left[\mathbf{f}(t, \mathbf{x})\right]_j\big|_{\mathbf{x}=\boldsymbol{\mu}}\mathbf{C}_{lp}}
\end{equation*}
where $j = 1,\dots, k$. The MS has a larger, $(k+k^2)$-dimensional state
space. It is convenient to consider this system as a network of $N+N^2$
coupled nodes. The $N$ components $\boldsymbol{\mu}_i(t)$ of
$\boldsymbol{\mu}(t)$ describe the mean values (over noise realizations) of the
stochastic network nodes, the $N$ diagonal $m_i\times m_i$ blocks of
$\mathbf{C}(t)$ -- variances, and the $N^2-N$ off-diagonal $m_i\times m_j$
blocks -- covariances between these nodes.  Hence, the MS generates a
connectivity graph $\Gamma^*$, which is related to the graph $\Gamma$ via
$\mathbf{f}$.

\section{\label{Sec:ControlNodes}Identifying network nodes for control}

For a deterministic network model, as given by \cref{eq:pt1:stochastic} without the noise term, i.e., 
\begin{equation} \label{eq:noiseless}
\dot{\mathbf{x}}=\mathbf{f}(t, \mathbf{x}),
\end{equation}
it is possible that only a subset of nodes needs to be pinned (i.e., prescribed
to a target trajectory) in order to drive the whole system to a target solution
(an attractor) \cite{Fiedler2013}. These nodes are termed \emph{switching
nodes} (SN) and formally defined as follows: 
\begin{definition}
    A subset $K$ of the vertex set $I:=\{1,\dots,N\}$ for the system given by
    \cref{eq:noiseless} is called a set of \emph{switching nodes}, if
    \begin{equation*}
        \lim_{t\to+\infty}\mathbf{x}_{I\setminus K}(t) = \mathbf{x}_{I\setminus K}^*(t)
    \end{equation*}
    holds for any solution $\mathbf{x}_{I\setminus K}(t)$ of the subsystem
    \begin{equation*}
        \dot{\mathbf{x}}_{I\setminus K} = \mathbf{f}_{I\setminus K}(t, \mathbf{x}_{I\setminus K}, \mathbf{x}^*_K)
    \end{equation*}
    and for any solution $\mathbf{x}^*(t)$ of \cref{eq:noiseless} (the full
    system).
    \label{def:sn}
\end{definition}
\noindent Here, $\mathbf{x}_{I\setminus K}(t)$ denotes the joint state vector
of all nodes in $I$ except the switching nodes $K$.  In other words, if the SN
are forced to attain the values from one of the solutions $\mathbf{x}^*(t)$ of
the network system, the whole network is guaranteed to converge to that
solution as $t\to\infty$ for all initial conditions.  It should be noted that
SN are a special case of the so-called \emph{determining nodes} proposed in
\cite{Fiedler2013}.  The difference is that the latter must (at least) tend to
$\mathbf{x}_K^*(t)$ as $t\to\infty$, whereas the former are forced to follow
$\mathbf{x}_K^*(t)$ exactly.  We use SN in this work because they are powerful
enough for the control purposes and simplify the mathematical derivations
below. 

Importantly, the set of SN can be identified based solely on the structural
information of the network [i.e., knowing $\Gamma$ and not necessarily
$\mathbf{f}(t, \mathbf{x})$].  The following assumptions are required:
\begin{assumption}
    \label{asm:dissip}
    Dissipativity: The system without noise, \cref{eq:noiseless} is
    dissipative, i.e., as time tends to infinity, the state of the system stays
    within a ball of finite radius, for all initial conditions.
\end{assumption}
\begin{assumption}
    \label{asm:decay}
    Decay condition:
    $\nabla_{\mathbf{x}_i}\mathbf{f}_i(t, \mathbf{x}_i, \mathbf{x}_{I_i}) < 0$,
    $\forall i \in \lbrace 1,\dots, N \rbrace$ and for all $t$ and bounded
    $\mathbf{x}$.  In other words, the \emph{Jacobian matrix} of $\mathbf{f}_i$
    with respect to $\mathbf{x}_i$ is strictly negative definite. 
\end{assumption}
If \cref{asm:decay} does not hold for a node $\mathbf{x}_j$, an explicit edge
from the node to itself is implied in the graph $\Gamma$ (see Section 1 in
\cite{Fiedler2013}).  With these two assumptions, the set of SN is identical to
the \emph{feedback vertex set} (FVS) of $\Gamma$, defined as:
\begin{definition}
    A \emph{feedback vertex set} of a directed graph $\Gamma$ is a possibly
    empty subset $K$ of vertices such that the subgraph $\Gamma\setminus K$ is
    acyclic.
    \label{def:fvs}
\end{definition}
\noindent Here, $\Gamma\setminus K$ denotes the subgraph that remains when all
vertices of $K$ are removed from $\Gamma$ along with all edges from or towards
those vertices.  This implies that the structural information alone can provide
clues about the dynamics and controllability of the system. Note, however, that
the problem of finding minimal FVS of a directed graph is NP-complete
\cite{Karp1972}. 

Let us now consider the graph $\Gamma^*$ of the MS, \cref{eq:pt1:mu,eq:pt1:c},
which can be revealed by rewriting that system using block-wise notation,
\begin{subequations}
    \begin{align}
        \dot{\boldsymbol{\mu}}_i &= \mathbf{f}_i+\frac{1}{2}\sum_{k, j\in\{i\}\cup I_i}\nabla_j\nabla_k \mathbf{f}_i \cdot \mathbf{C}_{kj} \label{eq:pt1:muel} \\
        \dot{\mathbf{C}}_{ij} &= \nabla_i \mathbf{f}_i\mathbf{C}_{ij} + \mathbf{C}_{ij}^T\nabla_i \mathbf{f}_j^T + \mathbf{Q}_{ij} \nonumber \\
        & \quad + \sum_{k\in I_i}\nabla_k \mathbf{f}_i \mathbf{C}_{kj} + \sum_{k \in I_j}\mathbf{C}_{ik}^T\nabla_k \mathbf{f}_j^T, \label{eq:pt1:cel}
    \end{align}
    \label{eq:pt1:mucel}%
\end{subequations}
where $\mathbf{f}_i$ is evaluated at $\mathbf{x}=\boldsymbol{\mu}$, i.e.,
$\mathbf{f}_i = \mathbf{f}_i(t,\boldsymbol{\mu}_i,\boldsymbol{\mu}_{I_i})$ and
$\nabla_k \mathbf{f}_i$ is the Jacobian of $\mathbf{f}_i$ with respect to the state
vector $\mathbf{x}_k$, evaluated at $\mathbf{x}=\boldsymbol{\mu}$. For improved
readability the dependency of $\mathbf{f}_i$ and its derivatives on
$\boldsymbol{\mu}_i$ and $\boldsymbol{\mu}_{I_i}$ is implied and not indicated
separately.  Notably, the Jacobian of the right-hand side of \cref{eq:pt1:muel}, i.e.,
the derivative with respect to $\boldsymbol{\mu}_i$, is not guaranteed to be negative
definite -- even if the graph $\Gamma$ of the original system does not contain
cycles (i.e., the fact that the decay condition holds for $\mathbf{f}_i$ and $I_i$ does not
include $i$ itself, for all $i \in I$).  This implies that every node
$\boldsymbol{\mu}_i$ of $\Gamma^*$ must have a self-connecting edge.
Thus, the minimal FVS (\cref{def:fvs}) of $\Gamma^*$ includes at least all
nodes $\boldsymbol{\mu}_i$, cf. \cref{eq:pt1:muel}.  To provide a simple
example, in \cref{fig:graph:torus} the subgraph structure of the moments system
of a network of 4 nodes with circular topology (directed ring) 
is visualized. Note, that in this example the minimal FVS of
the original graph $\Gamma$ consists of only one node. The minimal FVS of
$\Gamma^*$ -- according to the two subgraphs -- contains at least 11 nodes:
$\boldsymbol{\mu}_1,\dots,\boldsymbol{\mu}_4$ plus one arbitrary row and one arbitrary column of
$\mathbf{C}$.

\begin{figure}[!htbp]
    \includegraphics[width=0.95\columnwidth]{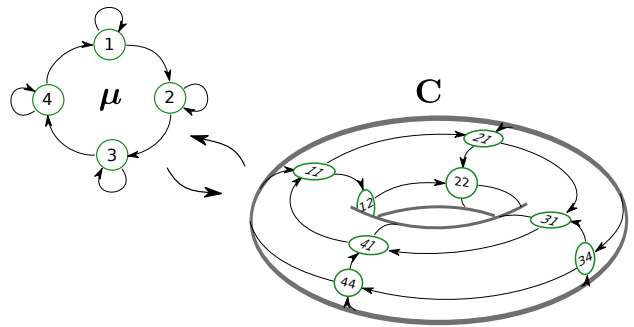}
    \caption{Subgraphs of connectivity graph $\Gamma^*$ of the MS (cf.
        \cref{eq:pt1:muel,eq:pt1:cel}) of a network (cf.
        \cref{eq:stochastic:components}) of $ N=4 $ nodes with circular graph
        $\Gamma$. Left: Subgraph for the connectivity of the expected value
        nodes $\boldsymbol{\mu}_i$.  Right: Subgraph for the connectivity
        between the (co)variance nodes $\mathbf{C}_{ij}$, $ i,j \in \lbrace
        1,\dots, N \rbrace $.  Numbers indicate the values of the indices of
        $\boldsymbol{\mu}$ and $\mathbf{C}$, respectively.
    }
    \label{fig:graph:torus}
\end{figure}

However, it is possible to find a set of SN that is smaller than the minimal
FVS of $\Gamma^*$. Using the dynamical properties of \cref{eq:pt1:mu,eq:pt1:c} together
with the structure of (any given) $\Gamma$ -- under the
\cref{asm:dissip,asm:decay}, and
\begin{assumption}
    \label{asm:noise}
    Weak noise: the probability
    $P(||\mathbf{x}(t)-\boldsymbol{\mu}(t)||<\epsilon)$ is close to $1$ for any
    norm $||\cdot ||$ and small $\epsilon >0$.
\end{assumption}
\noindent(already assumed for the method of moments) -- we obtain:

\begin{prop}
    Let $K\subseteq I$ be the minimal feedback vertex set of the graph $\Gamma$
    of the stochastic system, \cref{eq:pt1:stochastic}, with total vertex set
    $I:=\{1,\dots,N\}$, and let \cref{asm:dissip,asm:decay,asm:noise} hold.
    Then the set of switching nodes for the corresponding moments system,
    \cref{eq:pt1:mu,eq:pt1:c}, is given by 
    \begin{equation}
        K^* := \{\boldsymbol{\mu}_j, \mathbf{C}_{ij}, \mathbf{C}_{ji}: j \in K, i\in I \}.
        \label{eq:pt1:snmc}
    \end{equation}
    \label{prop:sn}
\end{prop}

\noindent In other words, if the FVS of the stochastic system is $K$, then the
set of SN $K^*$ of the associated MS consists of the means and variances of the
nodes in $K$ and, importantly, the covariances between these nodes and all
others (visualized in \cref{fig:moments} for $K=\lbrace2,3\rbrace$).  The set
of SN of the MS for the example ring network mentioned above consists of one
arbitrary mean vector $ \boldsymbol{\mu}_i $ plus the covariances between the
corresponding node ($ \mathbf{x}_i $) and all other nodes. The proof of
\cref{prop:sn} is given in the \cref{Sec:proofs}.

\begin{figure}[!htbp]
    \includegraphics[width=\columnwidth]{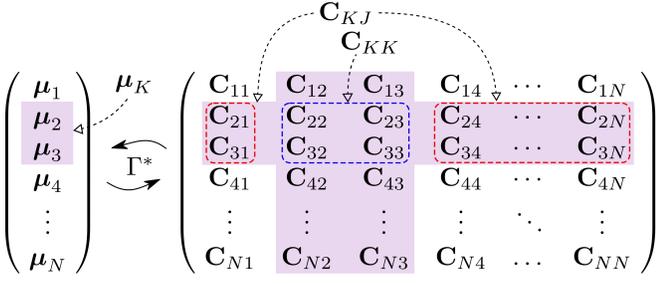}
    \caption{
        Schematic of mean vector $\boldsymbol{\mu}$ and covariance matrix
        $\mathbf{C}$, where a subset of nodes is highlighted by shaded
        background. That subset $K^*$ of nodes needs to be controlled, given
        that the FVS $K$ of the (original) network consists of nodes 2 and 3.
        In general, the set $K^*$ comprises the nodes $\boldsymbol{\mu}_i$ with
        the same index values as in the FVS of the original network, plus the
        corresponding rows and columns of the matrix $\mathbf{C}$, cf.
        \cref{prop:sn}.
        }
    \label{fig:moments}
\end{figure}

\section{\label{Sec:Controller}Designing the controller}

When controlled via pinning, the set of SN is capable to switch the entire
network system from any state to the desired metastable state.
According to \cref{prop:sn}, the set of SN $K^*$ of a MS includes the nodes
that describe the covariances between the nodes of the minimal FVS $K$ with the
rest of the nodes of the stochastic network system.  Assuming that the nodes in
$K$ are physically accessible, the question is how to choose the appropriate
$\mathbf{u}_K(t):=\mathbf{x}_K(t)$ for the stochastic system such that the set
of nodes $K^*$ of the corresponding MS are pinned (\cref{fig:control:full}).

Let us consider the system given by \cref{eq:pt1:stochastic}, with FVS $K$, as
two coupled subsystems,
\begin{subequations}
    \begin{align}
        \dot{\mathbf{x}}_K &= \mathbf{f}_K(t, \mathbf{x}_K, \mathbf{x}_J) + \mathbf{M}_K\boldsymbol{\eta}_K(t) \label{eq:pt2:xK} \\
        \dot{\mathbf{x}}_J &= \mathbf{f}_J(t, \mathbf{x}_J, \mathbf{x}_K) + \mathbf{M}_J\boldsymbol{\eta}_J(t), \label{eq:pt2:xJ}
    \end{align}
    \label{eq:pt2:split}%
\end{subequations}
where $J:= I \setminus K$, and $\boldsymbol{\eta}_{J}(t)$,
$\boldsymbol{\eta}_{K}(t)$ are independent Gaussian white noise processes.
Controlling the whole system is then reduced to controlling the subsystem,
\cref{eq:pt2:xJ}, with the control signal $\mathbf{u}_K(t)$ that is used to pin
the nodes in $ K $: $\mathbf{x}_K(t) = \mathbf{u}_K(t)$ and
\begin{equation}
    \dot{\mathbf{x}}_J = \mathbf{f}_J(t, \mathbf{x}_J, \mathbf{u}_K(t)) + \mathbf{M}_J\boldsymbol{\eta}_J(t).
    \label{eq:pt2:class1}
\end{equation}

Note that by definition of SN (\cref{def:sn}), if $\boldsymbol{\eta}_J(t)\equiv
0$, setting the control signal to the desired (deterministic) attractor,
$\mathbf{u}_K(t) = \mathbf{x}^*_K(t)$ guarantees that the subsystem
$\mathbf{x}_J$ necessarily converges to the desired state $\mathbf{x}^*_J(t)$.
With non-zero noise, however, we seek to guide the system to a metastable state in
the moments space, that is, we want to control the expected value and
covariance matrix of $\mathbf{x}(t)$ rather than the stochastic $\mathbf{x}(t)$
itself.  Therefore, the goal is to design a stochastic feedback controller
$\mathbf{u}_K(t, \mathbf{x}_J(t))$, such that the following conditions (in the
moments space) are satisfied:
\begin{subequations}
    \begin{align}
        &\langle \mathbf{u}_K \rangle(t) = \boldsymbol{\mu}_K^*(t) \label{eq:pt2:condmu} \\
        &\mathrm{cov}(\mathbf{u}_K, \mathbf{u}_K)(t) = \mathbf{C}^*_{KK}(t) \label{eq:pt2:condvar} \\
        &\mathrm{cov}(\mathbf{x}_J, \mathbf{u}_K)(t) = \mathbf{C}^*_{JK}(t) = \mathbf{C}_{KJ}^{*T}(t) \label{eq:pt2:condcov}
    \end{align}
    \label{eq:pt2:cond}%
\end{subequations}
where $\boldsymbol{\mu}^*(t)$ and $\mathbf{C}^*(t)$ denote the values of the target metastable state 
in the moments space.  We choose $\mathbf{u}_K(t, \mathbf{x}_J(t))$ as a linear
feedback controller of the form:
\begin{equation}
    \mathbf{u}_K(t, \mathbf{x}_J(t)) := [\boldsymbol{\mu}_g(t) + \mathbf{C}_g^{\frac{1}{2}}(t)\boldsymbol{\eta}_g(t)] + \mathbf{W}(t)^T\mathbf{x}_J(t)
    \label{eq:pt2:controller}
\end{equation}
where $[\boldsymbol{\mu}_g(t) +
\mathbf{C}_g^{\frac{1}{2}}(t)\boldsymbol{\eta}_g(t)]$ is an independent
time-uncorrelated Gaussian process with mean $\boldsymbol{\mu}_g(t)$ and
covariance $\mathbf{C}_g(t)$, and $\mathbf{W}(t) \in \mathbb{R}^{m_J\times
m_K}$ is a feedback weighting matrix to be determined. 
The expected value $\boldsymbol{\mu}_K(t) := \langle \mathbf{x}_K\rangle(t) =
\langle \mathbf{u}_K\rangle(t) $ is then given by
\begin{equation}
    \boldsymbol{\mu}_K(t) = \boldsymbol{\mu}_g(t) + \mathbf{W}^T(t) \boldsymbol{\mu}_J(t),
    \label{eq:pt2:muk}
\end{equation} 
the covariance matrix $\mathbf{C}_{KK}(t) := \mathrm{cov}(\mathbf{x}_K,
\mathbf{x}_K)(t) = \mathrm{cov}(\mathbf{u}_K, \mathbf{u}_K)(t)$ reads 
\begin{equation}
    \mathbf{C}_{KK}(t) = \mathbf{C}_g(t) + \mathbf{W}^T(t)\mathbf{C}_{JJ}(t)\mathbf{W}(t),
    \label{eq:pt2:var}
\end{equation}
and the cross-covariance matrix $\mathbf{C}_{JK}(t) :=
\mathrm{cov}(\mathbf{x}_J, \mathbf{x}_K)(t) = \mathrm{cov}(\mathbf{x}_J,
\mathbf{u}_K)(t)$ takes the form (time dependence of all variables is implied)
\begin{equation}
    \begin{split}
        \mathbf{C}_{JK}
        &= \langle \mathbf{x}_J\boldsymbol{\mu}_g^T\rangle + \langle \mathbf{x}_J \mathbf{x}_J^T \mathbf{W} \rangle - \boldsymbol{\mu}_J \boldsymbol{\mu}_g^T - \boldsymbol{\mu}_J \boldsymbol{\mu}_J^T \mathbf{W} \\
        &= \left(\langle \mathbf{x}_J\mathbf{x}_J^T\rangle - \boldsymbol{\mu}_J\boldsymbol{\mu}_J^T\right)\mathbf{W} = \mathbf{C}_{JJ}\mathbf{W}. 
    \end{split}
    \label{eq:pt2:cov}
\end{equation}
Applying condition (\ref{eq:pt2:condcov}) to \cref{eq:pt2:cov}, we obtain
\begin{equation}
    \mathbf{W}(t) = \mathbf{C}^{-1}_{JJ}(t)\mathbf{C}^*_{JK}(t),
    \label{eq:pt2:w}
\end{equation}
using condition (\ref{eq:pt2:condvar}) together with
\cref{eq:pt2:w,eq:pt2:var}, we obtain 
\begin{equation}
    \mathbf{C}_g(t) = \mathbf{C}_{KK}^*(t) - \mathbf{C}_{JK}^{*T}(t)\mathbf{C}_{JJ}^{-1}(t)\mathbf{C}_{JK}^*(t),
    \label{eq:pt2:cg}
\end{equation}
and, using condition (\ref{eq:pt2:condmu}), \cref{eq:pt2:w,eq:pt2:muk} yields
\begin{equation}
    \boldsymbol{\mu}_g(t) = \boldsymbol{\mu}_K^*(t) - \mathbf{C}_{JK}^{*T}(t)\mathbf{C}^{-1}_{JJ}(t)\boldsymbol{\mu}_J(t).
    \label{eq:pt2:mug}
\end{equation}
This means, applying the control signal $\mathbf{u}_K^*(t, \mathbf{x}_J(t))$
given by \cref{eq:pt2:controller} with time-varying parameters according to
\cref{eq:pt2:w,eq:pt2:cg,eq:pt2:mug} is equivalent to pinning the SN in the
moments space to the target state $(\boldsymbol{\mu}^*(t),
\mathbf{C}^*(t))$.  Note, that the control signal $\mathbf{u}_K^*$ depends on
the state vector $\mathbf{x}_J(t)$ (hence, full observability of the stochastic
system is required) as well as on the moments $\boldsymbol{\mu}_J(t)$ and
$\mathbf{C}_{JJ}(t)$

These moments can be calculated using an adapted MS, cf. \cref{eq:pt1:mu,eq:pt1:c},
that takes pinning into account: $\boldsymbol{\mu}_K(t) =
\boldsymbol{\mu}_K^*(t) $, $ \mathbf{C}_{KK}(t) = \mathbf{C}_{KK}^*(t) $, $
\mathbf{C}_{JK}(t) = \mathbf{C}_{KJ}^T(t) = \mathbf{C}_{JK}^*(t) $, and:
\begin{subequations}
	\begin{align}
	\dot{\boldsymbol{\mu}}_J &= \mathbf{f}_J + \frac{1}{2}\nabla\nabla \mathbf{f}_J \cdot \mathbf{C} \label{eq:pt2:muJ} \\
	\begin{split}
    \dot{\mathbf{C}}_{JJ} &= \nabla_J \mathbf{f}_J\mathbf{C}_{JJ} + \mathbf{C}_{JJ}\nabla_J \mathbf{f}_J^T + \mathbf{Q}_{JJ} \\
	&\quad + \nabla_K \mathbf{f}_J\mathbf{C}_{KJ} + \mathbf{C}_{KJ}\nabla_K \mathbf{f}_J^T,
	\end{split}\label{eq:pt2:cJ}
	\end{align}
	\label{eq:pt2:mucJ}%
\end{subequations}
where $ \mathbf{f}_J $ is evaluated at $ \mathbf{x} = \boldsymbol{\mu} $, i.e.,
$ \mathbf{f}_J = \mathbf{f}_J(t,\boldsymbol{\mu}_J,\boldsymbol{\mu}_K^*) $ and
$\nabla_K \mathbf{f}_J$ is the Jacobian of $\mathbf{f}_J$ with respect to the state
vector $\mathbf{x}_K$, evaluated at $\mathbf{x}=\boldsymbol{\mu}$.

Note that although the clamped MS possesses a globally stable target state (by
\cref{prop:sn}), it only approximates the true $\mathbf{C}_{JJ}(t)$. Therefore,
the proposed controller does not guarantee stability for $t\to\infty$, as the
errors can build up. In practice, however, our controller is capable of
switching between the metastable states and holding the system there for a long
time (see \cref{Sec:Results:sim,Sec:Results:sim2} for empirical evaluation).
\begin{figure}[!htbp]
    \includegraphics[width=0.8\columnwidth]{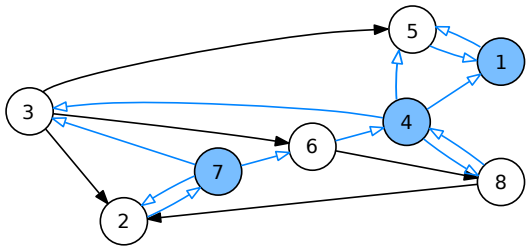}
    \caption{Connectivity graph $\Gamma$ of the example Hopfield network. 
        Vertices that belong to the selected minimal FVS $\{1, 4, 7\}$ as well
        as their incoming and outgoing edges are highlighted in blue.  There are no
        self-feedback loops because of the decay condition (\ref{asm:decay}).
        Note, that $\Gamma$ possesses a second minimal FVS of same cardinality:
        $\{1,2,4\}$.}
    \label{fig:graph}
\end{figure}

\begin{remark}
    In practice, when calculating $\mathbf{C}_g(t)$ according to
    \cref{eq:pt2:cg} it is possible that upon initiation of control the
    Cauchy-Schwarz inequality \cite{Tripathi1999}
    \begin{equation}
            \mathbf{C}_{KK}^*(t) \ge \mathbf{C}_{JK}^{*T}(t)\mathbf{C}_{JJ}^{-1}(t)\mathbf{C}_{JK}^*(t) 
            \label{eq:pt2:cauchy}
    \end{equation}
    is not satisfied, leading to negative definite $\mathbf{C}_g$.  This can
    occur because $\mathbf{C}_{KK}(t)$ and $\mathbf{C}_{JK}(t)$ are immediately
    forced to attain values on the target state, while the values of
    $\mathbf{C}_{JJ}$ still correspond to a distinct state.  In other words,
    the moments system (\cref{eq:pt2:muJ,eq:pt2:cJ}) might be dragged to a point, where
    it is not possible to construct a physical counterpart given by
    \cref{eq:pt2:class1} due to the constraints of probability laws.
    As a solution, we temporarily set $\mathbf{C}_g(t) \equiv 0$ in
    \cref{eq:pt2:controller} at the time points, where $\mathbf{C}_g(t)$ would
    otherwise be invalid.  For those time points $\mathbf{C}_{KK}(t)$ is no
    longer clamped to the target value $\mathbf{C}^*_{KK}(t)$. However, in this
    case, the MS still approximates the target state
    $(\boldsymbol{\mu}^*_J(t), \mathbf{C}^*_{JJ}(t))$, since the right-hand side of
    \cref{eq:pt2:cJ} only depends on $\mathbf{C}_{KK}$ via
    $\boldsymbol{\mu}_J(t)$, whose dynamics is dominated by $\mathbf{f}(t,
    \boldsymbol{\mu}_J, \boldsymbol{\mu}_K^*(t))$ by weak noise assumption.
    As the system approaches the target state, $\mathbf{C}_g(t)$ becomes
    well-defined again (since, by construction, the target state has
    physical meaning). In \cref{Sec:Results:sim,Sec:Results:sim2} we employ this approach using
    Monte-Carlo simulations. 
    \label{rem:cauchy}
\end{remark}

\section{\label{Sec:Results:sim}Example 1: Hopfield network}


We demonstrate the applicability of the theoretical results for general
stochastic dynamical networks from \cref{Sec:ControlNodes,Sec:Controller} using
an asymmetric Hopfield network \cite{Zheng2010}. These networks possess a
set of locally stable equilibrium states and a well-defined energy function
(see, e.g., \cite{Zheng2010, Hopfield2007, Hopfield1984, Rojas1996}). We
consider a stochastic Hopfield network described by
\begin{equation} 
    \dot{\mathbf{x}} = -\mathbf{x} + \mathbf{G}\boldsymbol{\phi}(\mathbf{x}) + \boldsymbol{\mu}_{\mathrm{ext}} + \sigma_{\mathrm{ext}}\boldsymbol{\eta}(t),
    \label{eq:sim:hopf}
\end{equation}
where $\mathbf{x}\in\mathbb{R}^N$ is the state vector, $\mathbf{G}$ the
(asymmetric) coupling matrix, $\boldsymbol{\phi}:\mathbb{R}^N\to\mathbb{R}^N$
denotes an element-wise hyperbolic tangent function, i.e.,
$\boldsymbol{\phi}(\mathbf{x}) = \left[\tanh(x_1),\dots,\tanh(x_N)\right]^T$,
$\boldsymbol{\mu}_{\mathrm{ext}}$ is the deterministic part of the external
input, $\boldsymbol{\eta}(t)$ is a $N$-dimensional Gaussian white noise
process, and $\sigma_{\mathrm{ext}}$ its (element-wise) standard deviation.
Here, $ N = 8 $ and $\mathbf{G}$ is given by the matrix
\begin{equation*}
    10^{-3} \cdot
    \begin{bmatrix}
        646 & 0 & 0 & 0 & 559 & 0 & 0 & 0 \\
        0 & 462 & 0 & 0 & 0 & 0 & -554 & 0 \\
        0 & -357 & 665 & 0 & 290 & -373 & 0 & 0 \\
        273 & 0 & 277 & 0 & 273 & 0 & 0 & 275 \\
        559 & 0 & 0 & 0 & 646 & 0 & 0 & 0 \\
        0 & 0 & 0 & -558 & 0 & 0 & 0 & -217 \\
        0 & -386 & 351 & 0 & 0 & 516 & 467 & 0 \\
        0 & 299 & 0 & 643 & 0 & 0 & 0 & 405
    \end{bmatrix}
    \label{eq:sim:G}
\end{equation*}
obtained using the algorithm described in \cite{Zheng2010}, where the weakest
edges were pruned to reduce the number of connections. 
The network graph (given
by $ \mathbf{G} $) with highlighted minimal FVS is shown in \cref{fig:graph}.
Note that the decay condition (\cref{asm:decay}) is fulfilled for all $ N $
nodes since the diagonal elements of $\mathbf{G}$ are all $ < 1 $. 
We further chose 
\begin{equation*}
 	\boldsymbol{\mu}_{\mathrm{ext}} = 10^{-2} \cdot
    \begin{bmatrix}
        -32& -7& 17& -36& 32& 16& -24& 3 \\
    \end{bmatrix}^T
    \label{eq:sim:Iinp}
\end{equation*}
and $\sigma_{\mathrm{ext}} = 0.01$ (larger values for $ \sigma_{\mathrm{ext}} $
were also considered, see below).
\begin{figure}[!htbp]
    \includegraphics[width=\columnwidth]{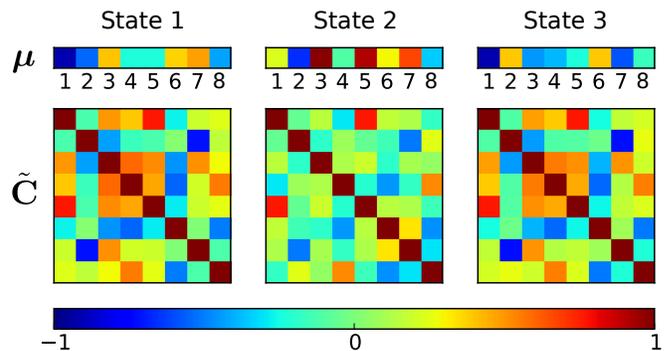}
    \caption{Metastable states of the example Hopfield network, given by
        \cref{eq:sim:hopf} ($\sigma_\mathrm{ext}=0.01$). Values of mean vector
        $\boldsymbol{\mu}$ and correlation matrix $\tilde{\mathbf{C}} =
        (diag(\mathbf{C}))^{-\frac{1}{2}}\mathbf{C}(diag(\mathbf{C}))^{-\frac{1}{2}}$,
        with covariance matrix $\mathbf{C}$, for the three states.  Note, that
    the three mean vectors are mutually distinct (approximately equidistant),
the correlation matrices of state 1 and 3, however, are similar, but not the
same.}
    \label{fig:fpoints}
\end{figure}

\begin{figure*}[!htbp]
    \includegraphics[width=0.9\linewidth]{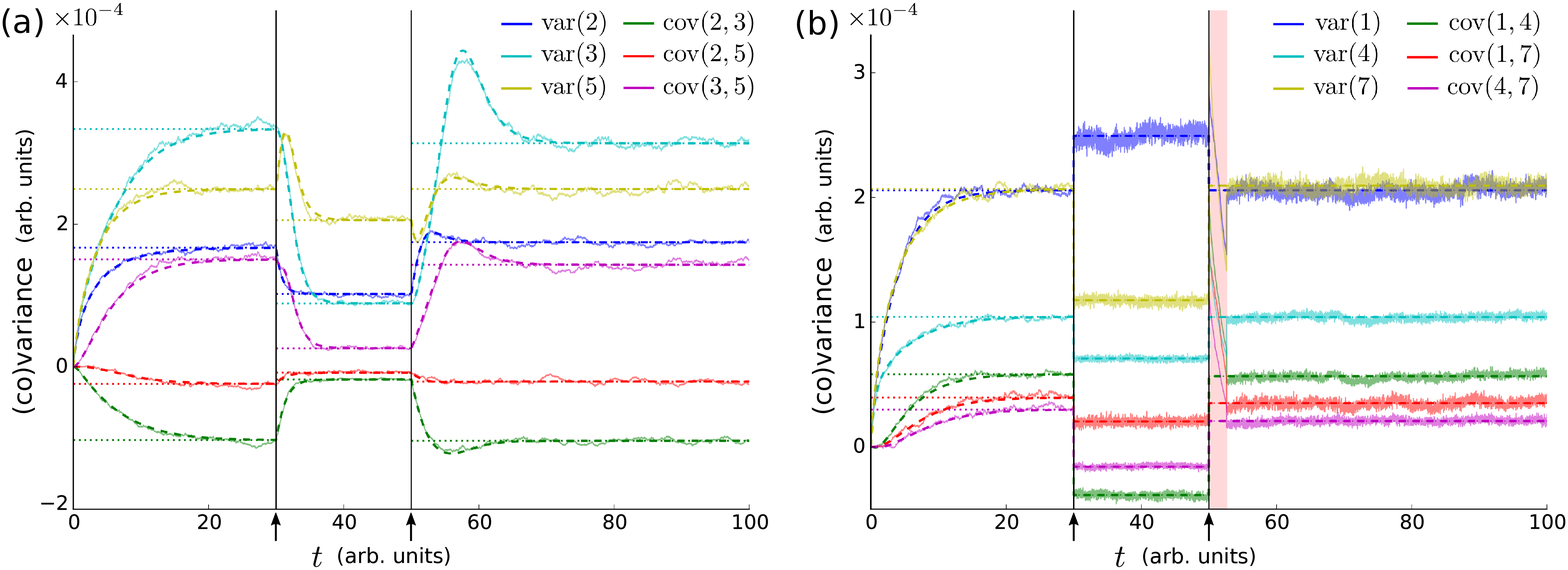}
    \caption{Switching states ($1\to2\to3$) of the Hopfield network ($\sigma_\mathrm{ext}=0.01$) by
        controlling the nodes of the FVS.  Time tracks of (co)variances of (a)
        nodes $\{2,3,5\}$ and (b) the pinned nodes of the FVS $\lbrace1,
        4, 7\rbrace$, obtained from Monte Carlo simulations (solid
        lines) and by solving the MS, \cref{eq:pt2:muJ,eq:pt2:cJ} (dashed lines).
        Dotted lines indicate the (co)variances values for states 1
        ($t\le30$), 2 ($30<t\le50$), and 3 ($t>50$). Arrows at $t=30$ and
        $t=50$ mark the initiation of control.  The shaded region starting at
        $t=50$ (right) indicates the period of time during which the noise
        intensity of the control signal is set to zero temporarily
        ($\mathbf{C}_g(t) = 0$) to prevent violation of the Cauchy-Schwarz
        inequality (see \cref{rem:cauchy}).  Note, that the Monte Carlo
        estimates of the pinned (co)variance values (i.e., those of
        $\mathbf{u}_K$) 
        exhibit increased high (temporal) 
        frequency because of the additive white Gaussian
        noise process of the control signal 
        (cf. \cref{eq:pt2:controller}).}
    \label{fig:control:fvs}
\end{figure*}
\begin{figure*}[!htb]
    \includegraphics[width=0.9\linewidth]{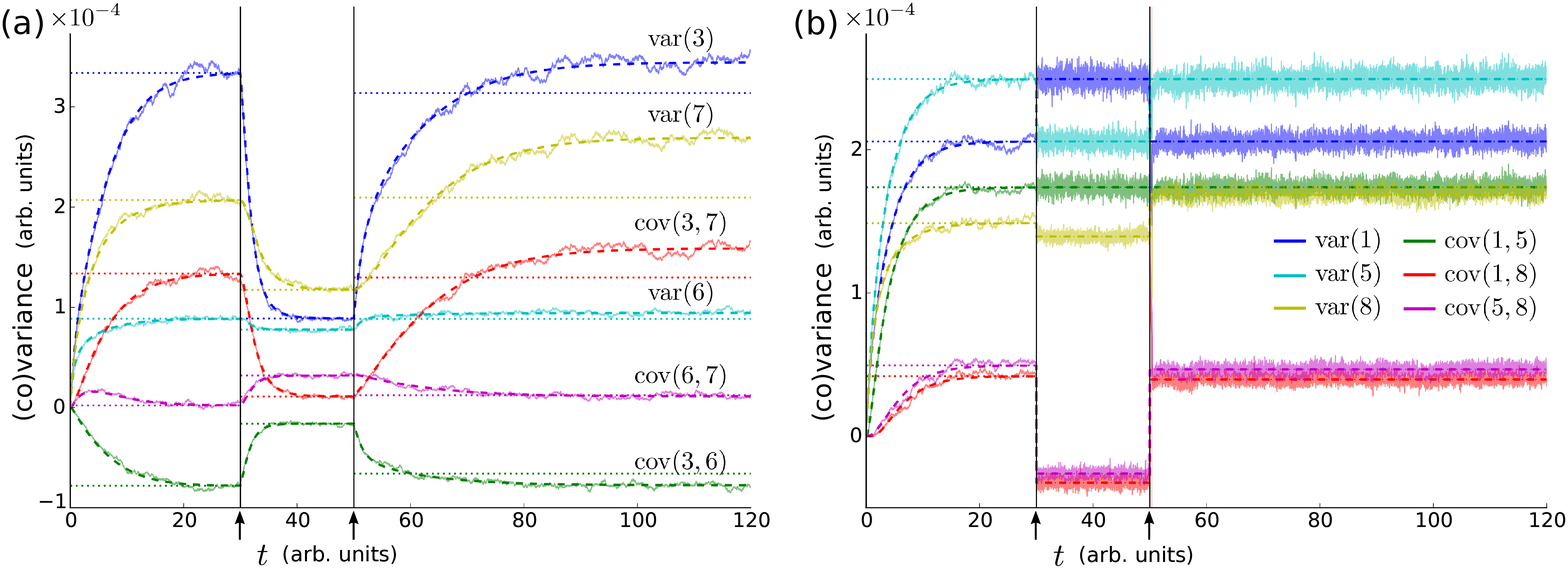}
    \caption{
        Switching states ($1\to2\to3$) of the Hopfield network ($\sigma_\mathrm{ext}=0.01$) by
        controlling the nodes $\lbrace1, 5, 8\rbrace$ (not a FVS).  Time tracks
        of (co)variances of (a) nodes $\lbrace3,6,7\rbrace$ and (b) the
        pinned nodes  $\{1,5,8\}$, obtained from Monte Carlo
        simulations (solid lines) and by solving the MS, \cref{eq:pt2:muJ,eq:pt2:cJ}
        (dashed lines).  Dotted lines indicate the (co)variances values for
        states 1 ($t\le30$), 2 ($30<t\le50$), and 3 ($t>50$). Arrows mark
        the initiation of control.  Note that switching $1\to2$ succeeds, but
        the controller fails to steer the system to state 3. 
    }
    \label{fig:control:wrong}
\end{figure*}

With the parameter values specified above the network without noise, i.e.,
$\sigma_{\mathrm{ext}} = 0$ in \cref{eq:sim:hopf}, possesses three stable fixed
point attractors with roughly equal Euclidean distance between each two. 
This gives rise to three metastable states when noise is present ($\sigma_{\mathrm{ext}}>0$).
That is, any solution trajectory of \cref{eq:sim:hopf} 
will fluctuate around the corresponding fixed point of the noiseless system 
for a certain period of time but eventually will escape 
its attraction domain and approach another fixed point, thereby transitioning between
metastable states. 
In other words, the solution of the Fokker-Planck equation that corresponds
to the Langevin equation (\ref{eq:sim:hopf}) will show a single peak close to one of the 
fixed points for a longer period of time -- if the probability density is initialized 
around that fixed point -- before the solution converges to the steady-state distribution
that reveals additional peaks (equivalent to the Boltzmann distribution).
The values of the mean vector $ \boldsymbol{\mu}(t) $ and
the covariance matrix $ \mathbf{C}(t) $ 
that correspond to these three metastable states are visualized in \cref{fig:fpoints}. 
%

All simulations were performed using Python with the libraries ``SciPy'' and
``Theano'' \cite{theano}.  The system of stochastic differential equations,
\cref{eq:sim:hopf}, was solved using a stochastic 4th order Runge-Kutta method
\cite{Kasdin1995} (integration step 0.01) and the corresponding MS given by a
set of ordinary differential equations was solved using the LSODA integration
method from ``SciPy'' (same integration step).

We used the controller described in \cref{Sec:Controller} to switch the
state of the system (in terms of mean and covariance values) from state 1 to
state 2 (at $t = 30$), and 
subsequently to state 3 (at $t=50$), by
pinning only the nodes of the FVS $K=\{1, 4, 7\}$. The system was initialized
near the state 1 and evolved freely until the control method was applied. 
Immediately after that moment the states of the pinned nodes changed abruptly according to the
control protocol. Specifically, given the target state $(\boldsymbol{\mu}^*(t),
\mathbf{C}^*(t))$ we first solved the adapted MS, \cref{eq:pt2:muJ,eq:pt2:cJ}.  Using
that solution we next calculated the parameters $\boldsymbol{\mu}_g(t)$,
$\mathbf{C}_g(t)$ and $\mathbf{W}(t)$ for the feedback control signal
$\mathbf{u}_K^*(t, \mathbf{x}_J(t))$ according to
\cref{eq:pt2:controller,eq:pt2:w,eq:pt2:cg,eq:pt2:mug}.  Finally, we simulated
the stochastic system 5000 times with the feedback controller
$\mathbf{u}_K^*(t,\mathbf{x}_J(t))$ for different noise realizations to
estimate the moments. The results of this procedure are shown in
\cref{fig:control:fvs,fig:control:wrong}. The values of $ \mathbf{C}(t) $
(subsets of $ \mathbf{C}_{JJ} (t)$ and $ \mathbf{C}_{KK} (t)$ shown in
\cref{fig:control:fvs,fig:control:wrong}) and the values of $
\boldsymbol{\mu}(t) $ (not shown) converge to the target values after a
transient transition period.
Note that the covariances between the clamped nodes ($ \mathbf{C}_{KK}(t) $)
may not follow the prescribed values immediately after pinning for that
transient period, cf. \cref{fig:control:fvs} (b), as explained in
\cref{rem:cauchy}. However, during the same period, the covariances between all
other nodes ($ \mathbf{C}_{JJ}(t) $) behave as predicted, cf.
\cref{fig:control:fvs} (a). 

\cref{fig:control:wrong} shows that choosing a subset of nodes that is not a
FVS (nodes $\lbrace 1, 5, 8 \rbrace$) does not guarantee controllability of
all metastable states. In particular, such controller fails to bring the system to
state 3.
\begin{figure}[!htbp]
    \includegraphics[width=0.9\columnwidth]{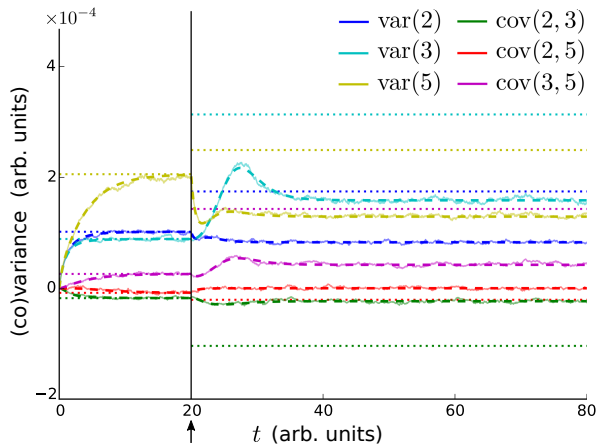}
    \caption{
        Switching states ($2\to3$) of the Hopfield network ($\sigma_\mathrm{ext}=0.01$) by controlling
        the nodes of the FVS $K=\{1, 4, 7\}$ using an open-loop controller,
        \cref{eq:sim:openloop}. Time tracks of (co)variances of nodes
        $\{2,3,5\}$, obtained from Monte Carlo simulations (solid lines) and by
        solving the MS, \cref{eq:pt2:muJ,eq:pt2:cJ} (with
        $\boldsymbol{\mu}_K(t)\equiv\boldsymbol{\mu}_K^*,
        \mathbf{C}_{KK}(t)\equiv \mathbf{C}_{KK}^*, \mathbf{C}_{JK}(t)\equiv
        0$, dashed lines).  Dotted lines indicate the (co)variances values for
        states 2 ($t\le20$) and 3 ($t>20$). The arrow marks the initiation
        of control.  Note, that the estimated (co)variance values (solid lines)
        follow the predictions from the MS solution (dashed lines), but do not
        converge to the target state. 
    }
    \label{fig:control:other:det}
\end{figure}

Finally, we demonstrate that it does not suffice to control the moments $
\boldsymbol{\mu}_K(t) $ and $ \mathbf{C}_{KK}(t) $, i.e., the mean and
covariance values among the nodes of the FVS, without controlling $
\mathbf{C}_{JK}(t) $ as well. To show this we use the simplified control signal 
\begin{equation}
    \mathbf{u}_K^*(t) = \boldsymbol{\mu}_K^*(t) + \mathbf{C}_{KK}^{*\frac{1}{2}}(t)\boldsymbol{\eta}_g(t) 
    \label{eq:sim:openloop}
\end{equation}
instead of $\mathbf{u}_K^*(t, \mathbf{x}_J(t))$ from \cref{eq:pt2:controller}.
This leads to $\mathbf{C}_{JK}(t) = 0$, since $\boldsymbol{\eta}_J(t)$ and
$\boldsymbol{\eta}_g(t)$ are independent stochastic processes.
As a consequence, the covariances $ \mathbf{C}_{JJ}(t) $ are not guaranteed to
converge to the target state (see \cref{fig:control:other:det}). 

We repeated the calculations for \cref{fig:control:fvs,fig:control:wrong} using
the same network system but with increased noise intensities ($\sigma=0.05,
0.1$). 
For $\sigma=0.05$, the controller was able to switch the network state and
hold it at the target, but the steady-state covariance values
deviated slightly from their predicted values (not shown).  This deterioration
is caused by the (generally) decreased approximation quality of the method of
moments for increased noise amplitudes. For $\sigma=0.1$ the MS becomes
unstable (i.e., it looses its dissipativity) which precludes computation of the
control signal.  Therefore, in practice, the stability of the MS may be used to
determine whether the noise is weak enough for successful application of the
control method. 

\section{\label{Sec:Results:sim2}Example 2: brain network}

To demonstrate the applicability of our theoretical results to larger systems we now consider a stochastic whole brain network model of 66 nodes which has been calibrated based on human diffusion imaging and functional magnetic resonance imaging data at rest, as described in \cite{Deco2013JN}.   
Each node represents the pooled neuronal activation of a brain region and is governed by a nonlinear mean field model that consists of a stochastic scalar differential equation with additive noise (Eqs.~(6)--(8) in \cite{Deco2013JN}). The covariance matrices of the node variables indicate functional connectivity between brain regions.
We use the same parameter values for the mean field model as in \cite{Deco2013JN}, except for the kinetic parameter $ \gamma $ which is changed to $\gamma =
0.241 \cdot 10^{-3}$ in order to meet the \cref{asm:decay}. 
The network graph is given by a thresholded asymmetric structural
connectivity matrix which was obtained using tractography based on diffusion (spectrum) imaging data and a partitioning of the brain into 66 regions, as described in \cite{Hagmann2008}. The resulting weighted graph is visualized in \cref{fig:control:brain} (a).
\cref{fig:control:brain} (b) shows the bifurcation diagram of the network states (in the absence of noise) as the global coupling parameter $ G $ varies.
For small values $ G $ only one stable state exists, characterized by a low neuronal firing activity in all brain regions. As $ G $ increases a first bifurcation emerges at a critical value, where multiple stable states of high activity appear and co-exist with the low activity state. As $ G $ further increases a second bifurcation appears where the low activity state becomes unstable.
This bifurcation diagram qualitatively resembles the one in \cite{Deco2013JN} (Fig. 2). 

\begin{figure*}[!htbp]
    \includegraphics[width=0.98\linewidth]{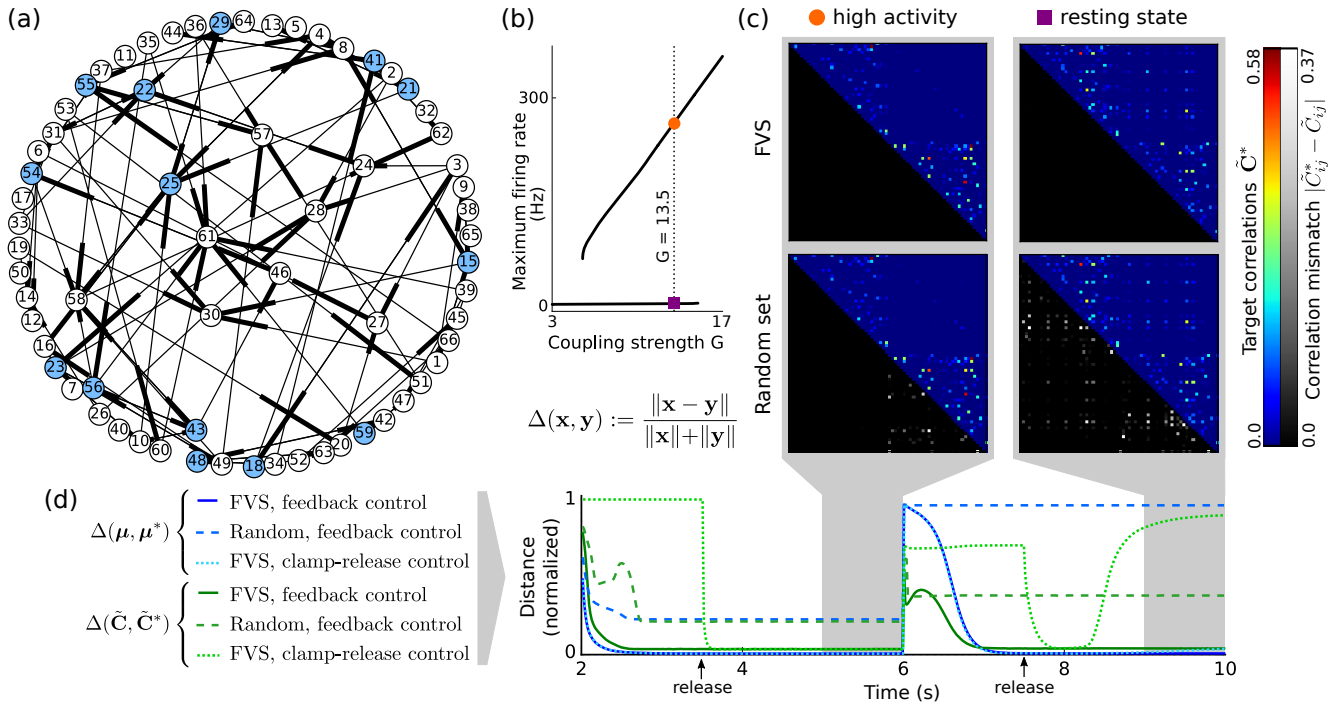}
    \caption{Brain network model. (a) Visualization of the weighted 66-node network graph. 
    	Nodes of the selected minimal FVS are highlighted in blue.
        Thick ends of the edges denote edge destination and strength. (b) Maximum
        firing rate of the network as a function of the global coupling
        strength $G$. Colored markers denote two different target states. (c)
        Correlation matrices at the target states (upper triangles) and their
        mismatch from the actual result of switching control (lower triangles).
        Rows correspond to two choices of nodes for the feedback control method: minimal FVS and a
        random set. Columns correspond to the target states indicated in (a). Correlation
        matrices are obtained by temporally averaging results of Monte-Carlo
        simulations ($ 2\cdot 10^5 $ iterations) over 1~s after the system has
        converged. (d) Normalized distance $\Delta$ between the actual and the
        target states as a function of time, in terms of means and correlation
        matrices, for the feedback controller using FVS or random nodes, and for 
        the clamp-release control scheme using FVS nodes.
        Note that diagonal elements are not considered when computing
        distance between correlation matrices (see \cref{Sec:Results:sim2}).
        The system first evolves without control for 2~s. Then, the control
        method is applied to switch to the target high activity state and 4~s
        later it is adjusted to switch to the resting state. For clamp-release control, 
        release occurs 
        1.5~s after each initiation.
    }
    \label{fig:control:brain}
\end{figure*}

We chose $G=13.5$, where the stochastic network possesses a metastable state
of low activity (resting state) and a set of metastable states associated with
high neuronal firing rates (``evoked'' high activity states). We
apply our feedback control method (\cref{eq:pt2:controller}) 
to switch the system to desired target states by
directly interfering only with a minimal FVS $K$ of $14$ nodes, see
\cref{fig:control:brain} (c). In particular, the system was initialized close to
the resting state, switched to a target high activity state after 2~s and back
to the resting state 4~s later. For comparison we used, as an alternative, a 
``clamp-release'' control strategy, where the FVS nodes
were clamped deterministically to the target $\boldsymbol{\mu}^*$ for 1.5~s and
then released, anticipating that the covariances would settle on the corresponding attractor in the moments space. 
We further used the feedback control method with three sets of randomly chosen nodes of the same size as the FVS (shown for one set).
To assess the quality of control, we used the normalized distance
$\Delta(\mathbf{x}, \mathbf{y}) := \frac{\rVert \mathbf{x} - \mathbf{y}
\lVert}{\rVert \mathbf{x} \lVert + \rVert \mathbf{y} \lVert}$ between the
target state and the actual state of the system, where $\rVert \cdot \rVert$ is
the Frobenius norm (\cref{fig:control:brain} (d)).  Note that when computing
distances between correlation matrices we do not take into account elements on
the diagonal, because $[\tilde{C}_1]_{ii} = [\tilde{C}_2]_{ii} = 1$ for any two
correlation matrices $\tilde{\mathbf{C}}_1, \tilde{\mathbf{C}}_2$.
%
Controlling the FVS nodes with the feedback controller successfully switched
the system to the desired states after a relatively short period (i.e., both
$\Delta(\boldsymbol{\mu}(t), \boldsymbol{\mu}^*)$ and
$\Delta(\tilde{\mathbf{C}}(t), \tilde{\mathbf{C}}^*)$ vanish), while
controlling a randomly chosen subset of nodes did not suffice. The
clamp-release control scheme (applied to the FVS set) succeeded to bring the system
to the high activity state with the appropriate correlation patterns, but failed to 
return it to the resting state
($\tilde{\mathbf{C}}(t)$ briefly approached $\tilde{\mathbf{C}}^*$ after the
release of control but started to diverge in less than 1~s).

\section{\label{Sec:Discussion}Discussion and Conclusion}

In this contribution we have developed a control method that is capable of
switching between metastable states -- in terms of the first and second moments
-- of a large class of stochastic nonlinear dynamical networks by pinning only
a subset of nodes. These nodes are identified based on information about the
graph of the network.
Our results can be applied in neuroscience, as demonstrated in \cref{Sec:Results:sim2}, where the second moments of
neuronal activities are often used to quantify functional connectivity. 
Because of their generality they may also be applied to
a variety of problems in engineering or physics.

A central limitation of the presented approach is the assumption of weak noise
(\cref{asm:noise}), which is required for the applied method of moments and
stability considerations in the proof of \cref{prop:sn}. Here, ``weak'' means
that the probability density $ p(x,t) $ of $x(t)$ (across noise realizations)
should be concentrated around the expected value $ \langle x\rangle(t) $. 

The assumption of weak noise implies that the dynamics of the system are
dominated by the deterministic part.  It seems, therefore, tempting to steer
the stochastic system to the desired state by pinning its set of SN
deterministically to values on the corresponding target attractor of the
noiseless system, and then ``release'' the control, to let the second moments
settle corresponding to that attractor.
%
Note, that releasing the control is required since deterministic pinning yields
inappropriate correlation patterns, as explained in \cref{Sec:Results:sim}.
Releasing the control, however, can cause the system to rapidly transition to a different state (cf.~\cref{fig:control:brain} (d)) . 
We have demonstrated that the proposed feedback controller allows to 
hold desired metastable states for long periods, while maintaining
the appropriate (desired) covariances.

Another limitation of the approach, is that the proposed feedback control
signal, cf.~\cref{eq:pt2:controller}, requires i) full observability of the
network's state to create non-zero covariance between the control signal and
the network nodes, and ii) knowledge of the instantaneous moments to compute
the proper feedback transformation matrix required to keep clamped covariances
at their target values, cf. \crefrange{eq:pt2:w}{eq:pt2:mug}.
One way to obtain this data is to solve the moments system
(\cref{eq:pt2:muJ,eq:pt2:cJ}), which requires knowledge of the system's
dynamics function $\mathbf{f}$, which might not be available.  However, for
systems with metastable fixed point states and weak $K\to J$ coupling (i.e.,
$||\nabla_K\mathbf{f}_J(t, \mathbf{x}_J, \mathbf{x}_K)||<\varepsilon$ for all
$t, \mathbf{x}$ and small positive $\varepsilon$) this problem can be
circumvented by using time invariant parameters $\boldsymbol{\mu}_g,
\mathbf{C}_g$ and $\mathbf{W}$, which are chosen such that conditions given by
\crefrange{eq:pt2:condmu}{eq:pt2:condcov} hold when
$\boldsymbol{\mu}_J(t)\equiv\boldsymbol{\mu}_J^*$ and
$\mathbf{C}_{JJ}(t)\equiv\mathbf{C}_{JJ}^*$: 
\begin{subequations}
    \begin{align}
        &\boldsymbol{\mu}_g = \boldsymbol{\mu}_K^* - \mathbf{C}_{JK}^{*T}\mathbf{C}^{*-1}_{JJ}\boldsymbol{\mu}_J^* \\
        &\mathbf{C}_g =  \mathbf{C}_{KK}^* - \mathbf{C}_{JK}^{*T}\mathbf{C}_{JJ}^{*-1}\mathbf{C}_{JK}^* \\
        &\mathbf{W} = \mathbf{C}_{JJ}^{*-1}\mathbf{C}_{JK}^* 
    \end{align}
    \label{eq:pt2:parstat}%
\end{subequations}
The weak coupling ensures that the full derivative $\nabla_J\mathbf{f}_J(t,
\mathbf{x}_J, \mathbf{u}_K(\mathbf{x}_J))$ is dominated by
$\nabla_J\mathbf{f}_J(\cdot, \mathbf{x}_J, \cdot)$ which guarantees the
existence of a global attractor in the controlled network.  In this way, the
control signal $\mathbf{u}_K(\mathbf{x}_J)$ is based only on knowledge of the
moments on the target state.

Finally, the FVS can be determined using a simple search procedure based on
\cref{def:fvs} (see Sec.~7.2 in \cite{Fiedler2013}). For very large graphs,
however, identification of the minimal FVS becomes computationally intractable,
as the corresponding optimization problem is NP-complete \cite{Karp1972}.
Therefore, an interesting problem would be to find approximate minimal feedback
vertex sets for large graphs at a reasonable computational cost.

\section{Acknowledgment}
We thank Lutz Schimansky-Geier for helpful comments on the manuscript.
This work was supported by DFG in the framework of collaborative research center SFB910.

\appendix*
\renewcommand{\theprop}{A.\arabic{prop}}
\section{\label{Sec:proofs}Proof of \cref{prop:sn}}

In order to prove \cref{prop:sn} we first provide three useful propositions.
Consider the deterministic system, \cref{eq:noiseless}, in terms of two coupled
subsystems,
\begin{subequations}
    \begin{align}
        \dot{\mathbf{x}}_K &= \mathbf{f}_K(t, \mathbf{x}_K, \mathbf{x}_J) \label{eq:a:simple1} \\
        \dot{\mathbf{x}}_J &= \mathbf{f}_J(t, \mathbf{x}_J, \mathbf{x}_K) \label{eq:a:simple2},
    \end{align}
    \label{eq:a:subsimple}%
\end{subequations}
cf. \cref{eq:pt2:xK,eq:pt2:xJ} without the noise terms. We define $\mathbf{g}^*(t,
\cdot):=\mathbf{f}_J(t, \cdot, \mathbf{x}^*_K)$ for a particular solution
$\mathbf{x}^*(t)$ of \cref{eq:noiseless}, such that
\begin{equation}
    \dot{\mathbf{x}}_J = \mathbf{g}^*(t, \mathbf{x}_J) \label{eq:a:clamped_subsystem}
\end{equation}
describes the dynamics of the state vector $\mathbf{x}_J$ where the time tracks
of the nodes $K$ are prescribed to $\mathbf{x}^*_K(t)$. 

In the following proposition we express the dynamics of the difference of any
two solutions of \cref{eq:a:clamped_subsystem} by a system of linear
differential equations with time-varying coefficients and relate the asymptotic
stability of that system to the controllability of the full system,
\cref{eq:a:simple1,eq:a:simple2}, via the set of nodes $K$.

\begin{prop}
    Let the linear system
    \begin{subequations}
    \begin{align}
        \dot{\mathbf{w}} &= \mathbf{A}(t)\mathbf{w}
	    \label{eq:a:var} \\
        \mathbf{A}(t) &:= \int_0^1 \nabla_{\mathbf{y}} \mathbf{g}^*(t, \mathbf{y}(t) + \nu[\tilde{\mathbf{y}}(t)-\mathbf{y}(t)])d\nu
	    \label{eq:a:A}
    \end{align}
    \label{eq:a:varsys}%
    \end{subequations}
    be globally asymptotically stable for any two solutions
    $\tilde{\mathbf{y}}(t)$, $\mathbf{y}(t)$ of 
    \begin{equation}
        \dot{\mathbf{y}} = \mathbf{g}^*(t, \mathbf{y})  
        \label{eq:a:simple2short}
    \end{equation}
    and any solution $\mathbf{x}^*(t)$ of \cref{eq:noiseless}. Then $K$ is a
    set of switching nodes (\cref{def:sn}) of the system \eqref{eq:noiseless}.

    \label{prop:snstab}
\end{prop}

\begin{proof}
    Observe that
    \begin{equation}
        \frac{d}{dt}(\tilde{\mathbf{y}}-\mathbf{y}) = \mathbf{g}^*(t, \tilde{\mathbf{y}}) - \mathbf{g}^*(t, \mathbf{y}) = \mathbf{A}(t)(\tilde{\mathbf{y}}-\mathbf{y}),
        \label{eq:a:difference}
    \end{equation}
    where we have used the Newton-Leibnitz formula for the equation on the
    right-hand side.  Due to the stability assumption on \cref{eq:a:var},
    the difference between any two solutions $\mathbf{y}$, $\tilde{\mathbf{y}}$
    of \cref{eq:a:simple2short} vanishes asymptotically as $t\to\infty$.
    Consequently, \cref{eq:a:simple2short} possesses a unique globally
    attractive solution. By construction, for any solution $\mathbf{x}^*(t)$ of
    \cref{eq:noiseless}, $\mathbf{x}^*_J(t)$ is the globally attractive
    solution of \cref{eq:a:simple2short}.  We obtain: $
    \mathbf{x}_K(t)=\mathbf{x}^*_K(t) \implies \mathbf{x}_J(t)\to
    \mathbf{x}^*_J(t)$ as $ t\to \infty $.  Thus, by \cref{def:sn}, $K$ is a
    set of switching nodes.
\end{proof}

In the following, the notation $\mathbf{A} < \mathbf{B} < \mathbf{0}$ means
that the matrices $\mathbf{A}$, $\mathbf{B}$ and $\mathbf{A}-\mathbf{B}$ are
(strictly) negative definite, i.e. $\mathbf{x}^T\mathbf{A}\mathbf{x} < 0,
\forall \mathbf{x} \ne \mathbf{0}$.

\begin{prop}
    Let $\mathbf{A}(t): \mathbb{R}\to\mathbb{R}^{k\times k}, k=\sum_{i=1}^Nm_i$
    be a continuous and uniformly bounded block-triangular matrix function,
    with $N$ blocks $\mathbf{A}_i(t)$, $ i \in \lbrace 1, \dots, N \rbrace $ on
    the main diagonal, and let each block $\mathbf{A}_i(t)$ be negative
    definite and bounded from above, i.e., $\mathbf{A}_i(t) <
    \bar{\mathbf{A}}_i < \mathbf{0}$ for all $t\ge 0$. Then, for any continuous
    and uniformly bounded matrix function $\mathbf{B}(t)$, there exists an
    $\varepsilon>0$, such that the (trivial) solution $\mathbf{z}(t)\equiv
    \mathbf{0}$ of the system 
    \begin{equation}
        \dot{\mathbf{z}} = [\mathbf{A}(t) + \varepsilon \mathbf{B}(t)]\mathbf{z}
	    \label{eq:a:perturbed}
    \end{equation}
    is globally asymptotically stable.

    \label{prop:perturb}
\end{prop}

\begin{proof}
    We first transform the unperturbed system
    \begin{equation}
        \dot{\mathbf{z}} = \mathbf{A}(t)\mathbf{z},
        \label{eq:a:unperturbed}
    \end{equation}
    to an equivalent system $ \dot{\mathbf{y}} =
    \tilde{\mathbf{A}}(t)\mathbf{y} $ with triangular matrix $
    \tilde{\mathbf{A}}(t) $.  The fundamental solution matrix $\mathbf{Z}(t)$
    of \cref{eq:a:unperturbed} can be uniquely \emph{QL-factorized},
    $\mathbf{Z}(t) = \mathbf{Q}(t)\mathbf{L}(t)$, where $\mathbf{Q}(t)$ is
    orthonormal and $\mathbf{L}(t)$ is a lower triangular matrix. Moreover,
    because $\mathbf{Z}(t)$ is lower block-triangular, the matrix
    $\mathbf{Q}(t)$ is block-diagonal with blocks $\mathbf{Z}_i(t) =
    \mathbf{Q}_i(t)\mathbf{L}_i(t)$. Using the change of basis
    $\mathbf{z}=\mathbf{Q}(t)\mathbf{y}$ we obtain (time dependence of
    $\mathbf{A}$ and $\mathbf{Q}$ is implied):
    \begin{equation}
        \dot{\mathbf{y}} = (\mathbf{Q}^T\mathbf{A}\mathbf{Q} - \mathbf{Q}^T\dot{\mathbf{Q}})\mathbf{y} =: \tilde{\mathbf{A}}(t)\mathbf{y},
        \label{eq:a:orthobasis}
    \end{equation}
    where $\tilde{\mathbf{A}}(t)$ is continuous, uniformly bounded and
    lower-triangular.  Note that since the diagonal blocks $\mathbf{A}_i(t)$ of
    $\mathbf{A}(t)$ are negative definite for all $t$, and $\mathbf{Q}(t)$ is
    block-diagonal with blocks $\mathbf{Q}_i(t)$, all diagonal entries of
    $\mathbf{Q}^T\mathbf{A}\mathbf{Q}$ are negative and bounded from above,
    i.e., for each diagonal block
    $\mathbf{D}_i:=\mathbf{Q}_i^T\mathbf{A}_i\mathbf{Q}_i$, we have $[D_i]_{jj}
    = \mathbf{q}_j^T(t)\mathbf{A}_i(t)\mathbf{q}_j(t) < \bar{a} < 0$, $\forall
    t$, where $i\in \lbrace 1,\dots, N\rbrace$, $j\in \lbrace 1,\dots,
    m_i\rbrace$ and $\mathbf{q}_j(t)$ is the $j$-th column of
    $\mathbf{Q}_i(t)$.  Moreover, the matrix $\mathbf{Q}^T\dot{\mathbf{Q}}$ is
    skew-symmetric and hence has zeros on the main diagonal (since
    $\mathbf{Q}^T\mathbf{Q} = \mathbf{I}$ and thus
    $d(\mathbf{Q}^T\mathbf{Q})/dt =
    \dot{\mathbf{Q}}^T\mathbf{Q}+\mathbf{Q}^T\dot{\mathbf{Q}} = \mathbf{0} $).
    It follows that the diagonal entries $\tilde{\mathbf{A}}(t)$ are
    continuous, negative and bounded from above.
    
    Using the same change of basis $\mathbf{z}=\mathbf{Q}(t)\mathbf{y}$ for
    \cref{eq:a:perturbed}, we obtain
    \begin{equation}
        \dot{\mathbf{y}} = [\tilde{\mathbf{A}}(t) + \varepsilon \tilde{\mathbf{B}}(t)]\mathbf{y},
        \label{eq:a:orthoperturbed}
    \end{equation}
    with $ \tilde{\mathbf{B}} := \mathbf{Q}^T\mathbf{B}\mathbf{Q}$.  According
    to Theorem 1.1 in \cite{Dai2006}, if the condition
    \begin{equation}
        \chi:=\limsup_{s\to+\infty}\frac{1}{T_s}\sum_{s'=0}^{s-1}\max_{1\le k\le N}\left\{ \int_{T_{s'}}^{T_{s'+1}}\tilde{A}_{kk}(t)dt \right\} < 0
        \label{eq:a:chi}
    \end{equation}
    is fulfilled, then for any given $\delta>0$ there exists an upper bound
    $\gamma$ on the norm of the perturbation $\varepsilon
    \tilde{\mathbf{B}}(t)$, i.e., $\varepsilon || \tilde{\mathbf{B}}(t)||<
    \gamma$ such that the \emph{maximum Lyapunov exponent} for the trivial
    solution of the system \eqref{eq:a:orthoperturbed} is bounded from above by
    $\chi + \delta$. Given that the diagonal elements of
    $\tilde{\mathbf{A}}(t)$ are negative and bounded from above (for all $ t
    $), the condition \eqref{eq:a:chi} holds for any partitioning
    $\{T_s\}_0^{\infty}$, and $\chi\le \bar{a}$.  Therefore, there exists a
    sufficiently small $\varepsilon > 0$ such that the system
    \cref{eq:a:orthoperturbed} is asymptotically and exponentially stable.

    Finally, observe that by Theorem 7.11 from \cite{meiss2007}, the
    characteristic exponents of the systems \eqref{eq:a:perturbed} and
    \eqref{eq:a:orthoperturbed} are identical.
\end{proof}

Next, we provide a generalization of the \emph{Kronecker sum} operation $
\oplus $ for partitioned matrices. 

\begin{prop}
    Let $\mathbf{A}, \mathbf{B} \in \mathbb{C}^{k\times k}, k=\sum_{i=1}^Nm_i$
    be two partitioned lower block-triangular matrices with $N$ blocks on the
    main diagonal.  Each diagonal block $\mathbf{A}_i$ and $\mathbf{B}_i$ is a
    square $m_i\times m_i$ matrix. All the elements above the diagonal blocks
    are zero. Then there exist a block diagonal \emph{permutation matrix}
    $\mathbf{P} \in \{0,1\}^{k^2\times k^2}$, such that the matrix
    \begin{equation}
        \mathbf{S} := \mathbf{P}(\mathbf{A}\oplus \mathbf{B})\mathbf{P}
        \label{eq:a:ksum}
    \end{equation}
    is again block-triangular, with $N^2$ blocks on the main diagonal, given by
    \begin{equation}
        \mathbf{S}_{i,j} := \mathbf{S}_{i(N-1)+j} = \mathbf{B}_j \oplus \mathbf{A}_i,
        \label{eq:a:ksum:block}
    \end{equation}
    where  $i,j\in \lbrace 1,\dots, N\rbrace$. 
    \label{prop:ksum}
\end{prop}

\begin{proof}
    By definition, the matrix $\mathbf{A}\oplus \mathbf{B}$ is given by
    \begin{equation}
        \mathbf{S}' := \mathbf{A} \oplus \mathbf{B} = \mathbf{A}\otimes \mathbf{I} + \mathbf{I} \otimes \mathbf{B},
        \label{eq:a:ksum:def}
    \end{equation}
    where $ \otimes $ is the Kronecker product and $\mathbf{I}$ is the $k\times
    k$ identity matrix.  $\mathbf{S}'$ has $N$ blocks $ \mathbf{S}'_i $, $ i
    \in \lbrace 1, \dots, N \rbrace $, of dimension $m_ik \times m_ik$  on the
    main diagonal, $ \mathbf{S}'_i = \mathbf{A}_i\otimes \mathbf{I} +
    \mathbf{I}_i \otimes \mathbf{B} $, where $\mathbf{I}_i$ is the $m_i \times
    m_i$ identity matrix.
    We next design a block-wise permutation matrix $\mathbf{P}$ such that each
    block $\mathbf{S}'_i$ itself has a block structure.  To that end, we define
    a block-diagonal matrix $\mathbf{P}$ with the same block partitioning of
    the main diagonal as the matrix $\mathbf{S}'$. Let each block
    $\mathbf{P}_i$ of $\mathbf{P}$ be a \emph{commutation matrix}, such that $
    \mathbf{P}_i(\mathbf{D}\otimes \mathbf{E})\mathbf{P}_i = \mathbf{E}\otimes
    \mathbf{D} $ for any two matrices $\mathbf{D}$ and $\mathbf{E}$ of
    dimensions $m_i\times m_i$ and $k \times k$, respectively.
    Then the blocks $ \mathbf{S}_i $ on the main diagonal of
    $\mathbf{S}:=\mathbf{P}\mathbf{S}'\mathbf{P}$ are given by
    \begin{equation}
        \mathbf{S}_i = \mathbf{P}_i\mathbf{S}'_i\mathbf{P}_i = \mathbf{I}\otimes \mathbf{A}_i + \mathbf{B} \otimes \mathbf{I}_i.
        \label{eq:a:ksum:blockswap}
    \end{equation}
    Note that each block $\mathbf{S}_i$ is block-triangular, and the $j$-th
    nested block $ \mathbf{S}_{i,j} $ on the main diagonal of $\mathbf{S}_i$
    reads
    \begin{equation}
        \mathbf{S}_{i,j} = \mathbf{I}_j \otimes \mathbf{A}_i + \mathbf{B}_j \otimes \mathbf{I}_i = \mathbf{B}_j \oplus \mathbf{A}_i,
        \label{eq:a:ksum:level2}
    \end{equation}
    for $i,j \in \lbrace 1,\dots,N \rbrace$. Hence,
    $\mathbf{P}(\mathbf{A}\oplus \mathbf{B})\mathbf{P}$ is a block-triangular
    matrix with $N^2$ blocks as given by
    \cref{eq:a:ksum:level2}.
\end{proof}

\begin{proof}[Proof of \cref{prop:sn}]
    Consider the network system, \cref{eq:pt1:stochastic}, as represented by
    the two coupled subsystems, \cref{eq:pt2:xK,eq:pt2:xJ}.  Let the set $J$ be
    topologically sorted according to the directed acyclic subgraph of the
    nodes in the set $I\setminus K$, such that node $J[1]$ has incoming edges
    only from the nodes in $K$, $J[2]$ -- only from the nodes in $K\cup J[1]$,
    etc.  With this ordering of the nodes represented by the state vector $
    \mathbf{x}_J $ (and the components $\mathbf{f}_J$, $ \mathbf{M}_J $, $
    \boldsymbol{\eta}_J $ in \cref{eq:pt2:xJ} ordered correspondingly) the
    Jacobian matrix $\nabla_{J}\mathbf{f}_J(t, \mathbf{x}_J, \mathbf{x}_K)$ is
    a lower block-triangular matrix, and each block on the main diagonal is
    strictly negative definite for all $t$ and bounded $\mathbf{x}$ due to the
    decay condition (\cref{asm:decay}).  Note, that for multivariate nodes ($
    m_i>1 $), the ordering affects $ \mathbf{x}_J $ ``block-wise":
    $\mathbf{x}_J = [\mathbf{x}_{J[1]}^T,
    \mathbf{x}_{J[2]}^T,\dots,\mathbf{x}_{J[N-|K|]}^T]^T$.

    Let $\boldsymbol{\mu}_J(t) := \langle \mathbf{x}_J\rangle (t)$ and
    $\mathbf{C}_{JK}(t) := cov(\mathbf{x}_J,\mathbf{x}_K)(t) = \langle
    \mathbf{x}_J\mathbf{x}_K^T \rangle(t) -
    \boldsymbol{\mu}_J(t)\boldsymbol{\mu}_K^T(t)$ as used in
    \cref{Sec:Controller}. That is, the $m_i\times m_j$ submatrix
    $\mathbf{C}_{ij}(t)$ of $\mathbf{C}_{JJ}(t)$ is the covariance matrix
    between $\mathbf{x}_{J[i]}(t)$ and $\mathbf{x}_{J[j]}(t)$. $K^* := \lbrace
    \boldsymbol{\mu}_K, \mathbf{C}_{KK}, \mathbf{C}_{JK}, \mathbf{C}_{KJ}
    \rbrace$ then denotes the set of switching node candidates, cf.
    \cref{eq:pt1:snmc}.
    Using \cref{prop:snstab} we can now prove \cref{prop:sn} by showing the
    stability of the solution of the subsystem, \cref{eq:pt2:muJ,eq:pt2:cJ}, of the MS,
    \cref{eq:pt1:mu,eq:pt1:c}, where the moments $\boldsymbol{\mu}_K(t) =
    \boldsymbol{\mu}_K^*(t) $, $ \mathbf{C}_{KK}(t) = \mathbf{C}_{KK}^*(t) $, $
    \mathbf{C}_{JK}(t) = \mathbf{C}_{KJ}^T(t) = \mathbf{C}_{JK}^*(t) $ are
    pinned.
	
    It is convenient to express \cref{eq:pt2:muJ,eq:pt2:cJ} in compact form as
	\begin{equation}
        \dot{\boldsymbol{\gamma}}_J = \mathbf{h}_J(t, \boldsymbol{\gamma}_J, \boldsymbol{\gamma}_K), \label{eq:a:mucJcomp}
	\end{equation}
    where $\boldsymbol{\gamma}_J =
    [\mathbf{C}_{JJ[1]}^T,\dots,\mathbf{C}_{JJ[N-|K|]}^T,
    \boldsymbol{\mu}_J^T]^T \in \mathbb{R}^{k_J^2+k_J}$, $k_J := \sum_{i \in
    J}m_i$ contains $ \boldsymbol{\mu}_J $ and the columns of $ \mathbf{C}_{JJ}
    $.  Considering the Jacobian matrix $\mathbf{U} =
    \mathbf{U}(t,\boldsymbol{\gamma}_J, \boldsymbol{\gamma}_K)
    :=\nabla_{J}\mathbf{h}_J(t, \boldsymbol{\gamma}_J, \boldsymbol{\gamma}_K)$
    the following block structure can be identified:
    \begin{equation}
        \begin{split}
            \mathbf{U} &= \left[\begin{BMAT}{c.c}{c.c}
                    \nabla_{J}\mathbf{f}_J\oplus\nabla_{J}\mathbf{f}_J & \mathbf{B}_{1} \\
                    * & \nabla_{J}\mathbf{f}_J + \mathbf{B}_{2} 
            \end{BMAT}\right] \\
            &= \left[\begin{BMAT}{c.c}{c.c}
                    \mathbf{A}\oplus \mathbf{A} & 0 \\
                    * & \mathbf{A}
            \end{BMAT}\right] + 
            \left[\begin{BMAT}{c.c}{c.c}
                    0 & \mathbf{B}_1 \\
                    0 & \mathbf{B}_2
            \end{BMAT}\right] =: \mathbf{U}_L + \mathbf{B},
        \end{split}
        \label{eq:a:main}
    \end{equation}
    where $\mathbf{A} = \mathbf{A}(t,\boldsymbol{\gamma}_J,
    \boldsymbol{\gamma}_K):=\nabla_{J}\mathbf{f}_J(t,
    \boldsymbol{\mu}_J,\boldsymbol{\mu}_K)$ and $\mathbf{B}=
    \mathbf{B}(t,\boldsymbol{\gamma}_J, \boldsymbol{\gamma}_K) = \mathbf{B}(t,
    \boldsymbol{\mu}, \mathbf{C})$ is a bounded matrix function that is linear
    with respect to $\mathbf{C}(t)$ and involves second and third derivatives of
    $\mathbf{f}_J$.
    Note, that $\mathbf{U}_L = \mathbf{U}_L(t,\boldsymbol{\gamma}_J,
    \boldsymbol{\gamma}_K)$ is a lower block-triangular matrix.
    The $k_J\times k_J$ lower-right block of $\mathbf{U}_L$, i.e., $ \mathbf{A}
    $, is a lower block-triangular matrix with negative definite blocks
    $\mathbf{A}_i:=\nabla_{{J[i]}}\mathbf{f}_{J[i]}$ on the main diagonal, as
    was noted above.  Considering the $k_J^2\times k_J^2$ upper-left block of
    $\mathbf{U}_L$ we know, by \cref{prop:ksum}, that there exists a
    permutation matrix $\mathbf{P}$ such that the matrix
    $\mathbf{P}(\mathbf{A}\oplus \mathbf{A})\mathbf{P}$ is lower
    block-triangular with $N^2$ blocks, each given by $\mathbf{A}_j \oplus
    \mathbf{A}_i$. Furthermore, each of these blocks is negative definite,
    which can be seen by inspecting the hermitian parts $\mathbf{H}_{i,j} :=
    (\mathbf{A}_j \oplus \mathbf{A}_i) + (\mathbf{A}_j \oplus \mathbf{A}_i)^T$,
    since a matrix is negative definite iff its hermitian part is negative
    definite:
    \begin{equation}
        \begin{split}
            \mathbf{H}_{i,j} &= (\mathbf{A}_j \oplus \mathbf{A}_i) + (\mathbf{A}_j \oplus \mathbf{A}_i)^T \\
            &= \mathbf{A}_j\otimes \mathbf{I}_i + \mathbf{A}_j^T \otimes \mathbf{I}_i + \mathbf{I}_j \otimes \mathbf{A}_i + \mathbf{I}_j \otimes \mathbf{A}_i^T \\
            &= (\mathbf{A}_j+\mathbf{A}_j^T)\otimes \mathbf{I}_i + \mathbf{I}_j \otimes (\mathbf{A}_i + \mathbf{A}_i^T).
        \end{split}
        \label{eq:a:blocknd}
    \end{equation}
    So, $\mathbf{H}_{i,j}$ is a sum of two real symmetric matrices, both of
    which have only negative eigenvalues, hence $\mathbf{H}_{i,j}$ and
    consequently $\mathbf{A}_j \oplus \mathbf{A}_i$ are negative definite. As a
    result, using a suitable permutation matrix $\mathbf{V}$, the Jacobian
    $\mathbf{U}_L'(t,\boldsymbol{\gamma}_J,
    \boldsymbol{\gamma}_K):=\mathbf{V}\mathbf{U}_L(t,\boldsymbol{\gamma}_J,
    \boldsymbol{\gamma}_K)\mathbf{V}$ is a lower block-triangular matrix with
    $N+N^2$ negative definite blocks for all $t$ and bounded
    $\boldsymbol{\gamma}_J$, $\boldsymbol{\gamma}_K$. 
    Moreover, the element-wise integrated matrix
    \begin{equation}
        \hat{\mathbf{U}}_L'(t) := \int_0^1 \mathbf{U}_L'(t, \boldsymbol{\xi}(t, \nu),\boldsymbol{\gamma}_K(t))d\nu
        \label{eq:a:Uhat}
    \end{equation}
    preserves the aforementioned properties of $\mathbf{U}_L'$ for any bounded
    $\boldsymbol{\xi}(t, \nu) \in \mathbb{R}^{k_J + k_J^2}$ and
    $\boldsymbol{\gamma}_K(t)$.
    
    By \cref{prop:snstab}, we have that $K^*$ is indeed a set of SNs of
    \cref{eq:pt1:mu,eq:pt1:c}, if the matrix $\hat{\mathbf{U}}'(t) := \int
    \mathbf{V}\mathbf{U}\mathbf{V} =: \hat{\mathbf{U}}_L'(t) +
    \hat{\mathbf{B}}'(t)$ yields a globally stable linear nonautonomous system
    \begin{equation}
        \dot{\mathbf{w}} = \hat{\mathbf{U}}'(t)\mathbf{w}.
        \label{eq:a:varpert}
    \end{equation}
    
    By the assumption of weak noise (\cref{asm:noise})
    $\sup_{t\in\mathbb{R}^+}||\mathbf{C}(t)|| < \varsigma$ with small
    $\varsigma>0$.  Because the perturbation matrix $\mathbf{B}$ is linear in
    $\mathbf{C}(t)$, the norm of $\hat{\mathbf{B}}'(t)$ multiplicatively
    depends on $\varsigma$, and therefore can be arbitrary small, assuming
    bounded second and third derivatives of $\mathbf{f}_J$.  Hence, by
    \cref{prop:perturb}, there exists a sufficiently small $\varsigma$, such
    that the solution of \cref{eq:a:varpert} is globally asymptotically stable.
\end{proof}




%

\end{document}